\documentclass[journal,twoside,web]{ieeecolor}

\usepackage{generic}
\usepackage{textcomp}

\usepackage{lcsys}

\usepackage{amssymb}
\usepackage{amsmath,amsfonts}
\usepackage{mathtools}
\usepackage{dsfont}
\usepackage{algpseudocode}
\usepackage{array}
\usepackage{textcomp}
\usepackage{stfloats}
\usepackage{url}
\usepackage{verbatim}
\usepackage{graphicx}
\usepackage[hidelinks]{hyperref}
\usepackage{subcaption}
\usepackage[affil-it]{authblk}
\usepackage{fancyhdr}
\usepackage{bm}
\usepackage{enumerate}  
\usepackage{cite}
\usepackage[linesnumbered,ruled,vlined]{algorithm2e}
\SetKwInput{KwInput}{Input}
\SetKwInput{KwOutput}{Output}
\SetKw{KwAssert}{Assert}

\usepackage{amsthm}
\usepackage{setspace}
\newtheorem{theorem}{Theorem}
\newtheorem{lemma}{Lemma}
\newtheorem{proposition}{Proposition}

\newtheorem{remark}{Remark}
\newtheorem{assumption}{Assumption}

\usepackage{setspace}

\usepackage{tikz}
\usepackage[framemethod=TikZ]{mdframed}
\usetikzlibrary{shapes,arrows}
\usetikzlibrary{arrows,calc,positioning}

\tikzset{
	block/.style = {draw, rectangle,
		minimum height=1.2cm,
		minimum width=1.2cm},
	input/.style = {coordinate,node distance=1cm},
	output/.style = {coordinate,node distance=2cm},
	arrow/.style={draw, -latex,node distance=2cm},
	pinstyle/.style = {pin edge={latex-, black,node distance=1cm}},
	sum/.style = {draw, circle, node distance=1cm},
}

\definecolor{backgreen}{HTML}{E9F3DF}
\definecolor{backblue}{HTML}{DAE5EC}
\definecolor{backpurple}{HTML}{E5E0E8}
\definecolor{backorange}{HTML}{F2E9DF}
\definecolor{lightergray}{gray}{0.9}

\mdfdefinestyle{callout}{%
	linecolor=black,
	linewidth=1pt,%
	roundcorner=0pt,
	innertopmargin=7pt,
	innerbottommargin=7pt,
	innerrightmargin=7pt,
	innerleftmargin=7pt,
	leftmargin = 0pt,
	rightmargin = 0pt,
	backgroundcolor=lightergray
}

\usepackage{dsfont}

\usepackage{dirtytalk}

\usepackage{endnotes}

\definecolor{cadetblue}{rgb}{0.33, 0.41, 0.58}

\DeclareMathOperator{\argmin}{\mathrm{argmin}}

\DeclareMathOperator{\EV}{\mathds{E}}

\DeclareMathOperator{\col}{\mathrm{col}}
\DeclareMathOperator{\diag}{\mathrm{diag}}
\DeclareMathOperator{\rank}{\mathrm{rank}}

\DeclareMathAlphabet{\doublestruck}{U}{BOONDOX-ds}{m}{n}
\newcommand{\ones}[0]{\mathds{1}}
\newcommand{\zeros}[0]{\doublestruck{0}}

\newcommand{\R}[0]{\mathbb{R}}

\newcommand{\Rnn}[0]{\mathbb{R}_{\geq0}}

\newcommand{\Pcal}[0]{\mathcal{P}}

\usepackage{xstring}

\begin{document}

\def\BibTeX{{\rm B\kern-.05em{\sc i\kern-.025em b}\kern-.08em
    T\kern-.1667em\lower.7ex\hbox{E}\kern-.125emX}}
\markboth{\journalname, VOL. XX, NO. XX, XXXX 2017}
{Author \MakeLowercase{\textit{et al.}}: Preparation of Papers for IEEE Control Systems Letters (August 2022)}

\title{A Unified Family-optimal Solution to\\Covariance Intersection Problems with\\ Semidefinite Programming}

\author{
Leonardo Pedroso$^{a,b}$,
W.P.M.H. (Maurice) Heemels$^{a}$,
Pedro Batista$^{b}$
\thanks{$^{a}$ Control Systems Technology section, Department of Mechanical Engineering, Eindhoven University of Technology, The Netherlands (e-mail: \{l.pedroso,m.heemels\}@tue.nl)}
\thanks{$^{b}$ Institute for Systems and Robotics, Instituto Superior T\'ecnico, Universidade de Lisboa, Portugal (e-mail: pbatista@isr.tecnico.ulisboa.pt)}
}

\maketitle 
\thispagestyle{empty}

\begin{abstract}
Covariance intersection (CI) methods provide a principled approach to fusing estimates with unknown cross-correlations by minimizing a worst-case measure of uncertainty that is consistent with the available information. This paper shows that a generalized CI framework, called overlapping covariance intersection (OCI), unifies several existing CI formulations within a single optimization-based framework. This unification enables the characterization of family-optimal solutions for multiple CI variants, including standard CI and split covariance intersection (SCI), as solutions to a semidefinite program, for which efficient off-the-shelf solvers are available. When specialized to the corresponding settings, the proposed family-optimal solutions recover the state-of-the-art family-optimal solutions previously reported for CI and SCI. The resulting formulation facilitates the systematic design and real-time implementation of CI-based fusion methods in large-scale distributed estimation problems, such as cooperative localization.
\end{abstract}


\vspace{-0.2cm}

\section{Introduction}

Many recent engineering applications involve networks of agents operating cooperatively in a shared environment. Agents exchange information and collectively infer quantities of interest that support tasks such as navigation, coordination, and situational awareness. For example, in satellite mega-constellations, spacecraft combine absolute positioning information with relative observations obtained through GNSS sensors and inter-agent communication to estimate their states \cite{FergusonHow2003,PedrosoBatista2023DistributedEKF}. Similarly, in vehicle-to-everything (V2X) systems, autonomous vehicles integrate onboard sensing with information received from surrounding infrastructure to infer the states of nearby vehicles and pedestrians \cite{QiuQiuEtAl2019}.

As the number of agents and information links grows, keeping track of the correlation between all pieces of information across the network quickly becomes computationally prohibitive. Systems of this nature are classified as \emph{ultra large-scale systems}, for which centralized design and estimation strategies are computationally infeasible \cite{PedrosoBatistaEtAl2025ULSS}. Instead, distributed fusion approaches are necessary. This fact precludes the use of classical centralized filtering approaches that require full knowledge of all correlations \cite{AndersonMoore1979}. However, ignoring these correlations during fusion may result not only in degraded estimation accuracy, but also in \emph{inconsistent} uncertainty representations, i.e., agents may produce error covariance estimates that are unrealistically small relative to the true uncertainty. This is a well-known effect referred to as \emph{double-counting} \cite{PanzieriPascucciSetola2006,ChangChongEtAl2010}. Such inconsistency can have serious consequences in safety-critical applications, such as debris avoidance maneuvers for satellites. 

	To avoid double-counting, a family of methods known as \emph{covariance intersection} (CI) has been developed over the past quarter century \cite{ForslingNoackEtAl2024}. CI techniques provide a principled approach for combining estimates originating from multiple sources when the underlying cross-correlations are unknown or partially known. By construction, the fused estimate is guaranteed to remain consistent, meaning that the associated covariance matrix upper bounds the true estimation error covariance. The fusion weights are selected to minimize a worst-case uncertainty metric over all cross-correlations compatible with the locally available covariance information.

In this paper, we study a generalized covariance intersection formulation, called overlapping covariance intersection (OCI), which extends a problem originally introduced in \cite{PedrosoBatistaEtAl2026OCI}. The OCI framework admits a family-optimal solution that can be efficiently obtained as the solution of a semidefinite program (SDP), for which off-the-shelf solvers with polynomial worst-case complexity are available. A key insight of this work is that several existing CI-based fusion formulations can be naturally embedded within the OCI framework. As a consequence, their corresponding family-optimal solutions can be derived within a single unifying optimization-based approach and computed by solving an SDP. This unified characterization enables the systematic and computationally efficient implementation of CI-like fusion methods in ultra large-scale settings, including real-time cooperative localization.\vspace{-0.1cm}

\subsection{Contributions}

Specifically, the contributions of this paper are threefold:
\begin{itemize}
	\item We derive a family-optimal solution for an OCI problem, extending the formulation previously proposed in \cite{PedrosoBatistaEtAl2026OCI}.
	\item We show that the standard CI problem \cite{JulierUhlmann2017} and the split covariance intersection (SCI) problem \cite{LiNashashibiEtAl2013} can be cast as particular instances of the OCI framework.
	\item We express family-optimal solutions to CI and SCI problems with an arbitrary number of partial estimates using a generalization of the family commonly used in the CI literature. For the first time, we characterize such family-optimal solutions as the solution to an SDP. A numerical implementation with working examples is made available in an open-source repository.\vspace{-0.1cm}
\end{itemize}

\subsection{Notation}\vspace{-0.1cm}
Throughout this paper, the $n\times n$ identity, $n\times m$ null, and $n\times m$ ones matrices are denoted by $\mathbf{I}_n$, $\zeros_{n\times m}$, and $\ones_{n\times m}$, respectively. When clear from the context, the subscripts will be dropped to streamline notation.
The sets of $n\times n$ real symmetric positive semidefinite and positive definite matrices are denoted by $S^n_+$ and  $S^n_{++}$, respectively. Moreover, $\mathbf{P} \succ \zeros$ ($\mathbf{P}\succeq \zeros$) denotes that the symmetric matrix $\mathbf{P} \in \R^{n\times n}$ is positive definite (semidefinite) and $\mathbf{P} \succ \mathbf{Q}$ ($\mathbf{P} \succeq \mathbf{Q}$) denotes that the symmetric matrix $\mathbf{P}-\mathbf{Q} \in \R^{n\times n}$  is positive definite (semidefinite). Given a matrix $\mathbf{A}\in \R^{n\times m}$, $\mathbf{A}^+$ denotes the Moore-Penrose inverse of $\mathbf{A}$ \cite[Chap.~1.6]{Zhang2005}.

\section{Unifying Problem}\label{sec:oci}
\vspace{-0.1cm}
In this section, we present a generalized CI problem and a family-optimal solution that can be computed efficiently. The unifying CI problem is called overlapping covariance intersection (OCI). A particular case of this problem was recently introduced in \cite{PedrosoBatistaEtAl2026OCI}, motivated by the cooperative localization problem in emerging ultra-large scale systems.

\vspace{-0.1cm}
\subsection{Overlapping Covariance Intersection Problem}\label{sec:oci_probelm}

The goal is to estimate a state  $\mathbf{x} \in \R^n$ with information about $N$ partial estimates $\mathbf{z}_i = \mathbf{H}_i\mathbf{x}+\mathbf{e}_i$ with $i = 1,2,\ldots,N$. Each $\mathbf{H}_i$ is known and $\mathbf{e}_i$ is zero-mean random noise. One can concatenate all estimates as $\mathbf{z} := [\mathbf{z}_1^\top,\ldots,\mathbf{z}_N^\top]^\top \in \R^{o}$ and write $\mathbf{z} = \mathbf{H}\mathbf{x} + \mathbf{e}$, where $\mathbf{H} := [\mathbf{H}_1^\top \;\; \cdots \;\; \mathbf{H}_N^\top]^\top \in \R^{o\times n}$ and $\mathbf{e} := [\mathbf{e}_1^\top,\ldots,\mathbf{e}_N^\top]^\top  \in \R^{o}$. The second moment of $\mathbf{e}$ is denoted by $\EV[\mathbf{e}\mathbf{e}^\top]$. We want to design a linear fusion law to compute an unbiased estimator $\mathbf{\hat{x}} = \mathbf{K}\mathbf{z}$, where $\mathbf{K}\in \R^{n\times o}$ is the fusion gain. The estimate $\mathbf{\hat{x}}$ is unbiased if $\EV[\mathbf{\hat{x}}] = \mathbf{x}$, which is equivalent to imposing $\mathbf{K}\mathbf{H} = \mathbf{I}$ when designing $\mathbf{K}$. We consider that $\EV[\mathbf{e}\mathbf{e}^\top]$ is not exactly known and has the form \vspace{-0.1cm}
\begin{equation}\label{eq:E_ee}
	\EV[\mathbf{e}\mathbf{e}^\top] = \mathbf{R} + \mathbf{C}\mathbf{P}\mathbf{C}^\top, \vspace{-0.1cm}
\end{equation}%
where $\mathbf{R}\in S^o_{+}$ and $\mathbf{C} \in \R^{o\times m}$ are known and $\mathbf{P}\in S^m_{++}$ is not exactly known. We consider that there is information about $M$ bounds on components of $\mathbf{P}$. Each bound  $b \in \{1,2, \ldots,M\}$ is written as $\mathbf{W}_b\mathbf{P}\mathbf{W}_b^\top \preceq \mathbf{X}_b$, where $\mathbf{W}_b \in \R^{o_b\times m}$ and $\mathbf{X}_b \in S_{++}^{o_b}$.  Thus, the set of admissible matrices $\mathbf{P}$ is given by \vspace{-0.05cm}
\begin{equation}\label{eq:Pcal_def}
	\Pcal \!:=\! \{\mathbf{P} \!\in \!S_{++}^m \!:\! \mathbf{W}_b\mathbf{P}\mathbf{W}_b^\top \!\preceq \mathbf{X}_b \;\forall b \in\! \{1,2,\ldots,M\} \}.\vspace{-0.05cm}
\end{equation}

The estimation error is a random vector denoted by $\mathbf{\tilde{x}} := \mathbf{\hat{x}}-\mathbf{x}$. From the constraint $\mathbf{K}\mathbf{H} = \mathbf{I}$ on the gain, it follows that $\mathbf{\tilde{x}} =\mathbf{K}\mathbf{z}-\mathbf{x} = \mathbf{K}(\mathbf{H}\mathbf{x}+\mathbf{e})-\mathbf{x} = \mathbf{K}\mathbf{e}$ and $\EV[\mathbf{\tilde{x}}\mathbf{\tilde{x}}^\top] = \mathbf{K}	\EV[\mathbf{e}\mathbf{e}^\top] \mathbf{K}^\top$. The OCI problem is to design a gain $\mathbf{K}$ that optimizes an upper bound on the worst-case second moment of the estimation error for all admissible $\mathbf{P}\in \Pcal$. Formally, the goal is to solve the optimization problem
\begin{equation}\label{eq:OCI_orig_prob}
	\begin{aligned}
			&\min_{\substack{\mathbf{K}\in \R^{n\times o}, \mathbf{B}\in S^n_+}}  && J(\mathbf{B})\\
			&\quad\quad\; \mathrm{s.t.} &&  \mathbf{K}\mathbf{H} = \mathbf{I}\\  
			&&&	\mathbf{B} \succeq \mathbf{K}(\mathbf{R} + \mathbf{C}\mathbf{P}\mathbf{C}^\top)\mathbf{K}^\top,\; \forall \mathbf{P} \in \Pcal,
		\end{aligned}
\end{equation}
where $J: S^n_+ \to \R$ is any optimality criterion that satisfies the following monotonicity condition (which holds for common criteria in fusion applications such as the trace or determinant).

\begin{assumption}\label{ass:J}
	Given $\mathbf{X},\mathbf{Y} \in S^n_+$, the map $J: S^n_+ \to \R$ is such that $\mathbf{X} \succ \mathbf{Y} \implies J(\mathbf{X}) > J(\mathbf{Y})$.
\end{assumption}

The optimization problem~\eqref{eq:OCI_orig_prob} is a robust optimization problem. Problems of this form are studied in the context of linear unbiased estimation in \cite{ForslingHanssonEtAl2022} for particular structures of the ambiguity set $\Pcal$.

\subsection{Efficient Family-optimal Solution}\label{sec:oci_sol}

Computational efficiency is essential for applications of CI algorithms such as cooperative localization. However, it is computationally challenging to numerically solve the optimization problem \eqref{eq:OCI_orig_prob} because it is nonconvex and its second constraint is not supported by off-the-shelf solvers. To enable a computationally efficient solution, similarly to the literature on CI problems, we parameterize a well-behaved family of bounds for all $\mathbf{P}\! \in\! \Pcal$. The family is characterized in a way that the particularization of  \eqref{eq:OCI_orig_prob} to the family of bounds is computationally efficient and we call the corresponding solution a \emph{family-optimal} solution.

In this paper, we use a family of circumscribing ellipsoids studied in \cite{Kahan1968}, which is characterized by
\begin{equation}\label{eq:kahan_family}
	\textstyle\Pcal_{\mathrm{KF}}(\boldsymbol{\omega}):= \left\{\mathbf{P} \in S_{++}^m : \mathbf{P}^{-1} \succeq  \sum_{b=1}^M\boldsymbol{\omega}_b\mathbf{Y}_b\right\},
\end{equation}
where  $\mathbf{Y}_b := \mathbf{W}_b^\top\mathbf{X}_b^{-1} \mathbf{W}_b$ for $b =1,2,\ldots,M$ and $\boldsymbol{\omega} \in  \Delta^M := \{\boldsymbol{\omega} \in \Rnn^{M} : \ones^\top \boldsymbol{\omega} = 1 \}$ is a vector that parameterizes the family. Henceforth, we call this the \emph{Kahan family} of bounding ellipsoids. Remarkably, this family is a generalization of the most common family of bounding ellipsoids used for the standard CI problem, e.g., \cite{ReinhardtNoackEtAl2015}.  The following lemma establishes that $\Pcal_{\mathrm{KF}}(\boldsymbol{\omega})$ is indeed a bound on $\Pcal$ for all $\boldsymbol{\omega} \in  \Delta^M$.

\begin{lemma}\label{lem:kahan_conservative}
	The Kahan family defined in \eqref{eq:kahan_family} parameterizes bounds on $\Pcal$ defined in \eqref{eq:Pcal_def}, in the sense that $\Pcal \!\subseteq \!\Pcal_{\mathrm{KF}}(\boldsymbol{\omega})$ for all $\boldsymbol{\omega} \!\in\!  \Delta^M\!$.
\end{lemma}
\begin{proof}
	See Appendix~\ref{sec:proof_lem_kahan_conservative}.
\end{proof}

Incorporating the conservative information constraints into the OCI problem \eqref{eq:OCI_orig_prob} leads to the \emph{Kahan-family} OCI problem
\begin{equation}\label{eq:OCI_kahan}
	\begin{aligned}
			&\!\!\!\min_{\substack{\mathbf{K}\in \R^{n\!\times \!o}, \mathbf{B}\in S^n_+\\ \boldsymbol{\omega} \in \Delta^M}} && \!\!\!\! J(\mathbf{B})\\
			&\!\!\!\;\quad\quad \mathrm{s.t.} &&  \!\!\!\!  \mathbf{K}\mathbf{H} = \mathbf{I}\\  
			&&& \!\!\!\! 	\mathbf{B} \succeq \mathbf{K}(\mathbf{R} \!+\! \mathbf{C}\mathbf{P}\mathbf{C}^\top)\mathbf{K}^\top\!\!,\; \forall \mathbf{P} \in \Pcal_{\mathrm{KF}}(\boldsymbol{\omega}).
		\end{aligned}%
\end{equation}%
From Lemma~\ref{lem:kahan_conservative} one concludes that introducing the Kahan family in \eqref{eq:OCI_kahan} amounts to tightening the constraint $\mathbf{B}\! \succeq\! \mathbf{K}(\mathbf{R} + \mathbf{C}\mathbf{P}\mathbf{C}^\top)\mathbf{K}^\top,\, \forall \mathbf{P} \!\in\!\Pcal,$ in the original OCI problem \eqref{eq:OCI_orig_prob}. Notice that $M-1$ degrees of freedom have been introduced via $\boldsymbol{\omega}$ in \eqref{eq:OCI_kahan} to parameterize the family. Given a solution $(\mathbf{K}^\star,\mathbf{B}^\star,\boldsymbol{\omega}^\star)$ to \eqref{eq:OCI_kahan}, we say that the pair $(\mathbf{K}^\star,\mathbf{B}^\star)$ is the \emph{Kahan-family-optimal solution} to the OCI problem \eqref{eq:OCI_orig_prob}. 

In the next two results, we provide a computationally efficient characterization of the  \emph{Kahan-family-optimal solution}, i.e., a solution to optimization problem \eqref{eq:OCI_kahan}, for two separate cases: (i)~$\mathbf{R} = \zeros$, which is analyzed for the first time in this paper and does not follow immediately from the results or proofs in \cite{PedrosoBatistaEtAl2026OCI}; and (ii)~$\mathbf{R} \succ \zeros$, which was analyzed in \cite{PedrosoBatistaEtAl2026OCI}. The results corresponding to each of these cases will be used in Sections~\ref{sec:CI} and~\ref{sec:SCI} to characterize solutions to the CI and SCI problems.

Notice that, for the case $\mathbf{R} = \zeros$, one can rewrite \eqref{eq:E_ee} as $\EV[\mathbf{e}\mathbf{e}^\top] = \mathbf{C}\mathbf{P}\mathbf{C}^\top$. Therefore, $\EV[\mathbf{e}\mathbf{e}^\top]$ is made up uniquely of $o$ components of $\mathbf{P}\in S_{++}^m$ extracted by $\mathbf{C} \in \R^{o\times m}$. As a result, without any loss of generality, one can consider that $o = m$ and that $\mathbf{C}$ is full rank.

	\begin{theorem}\label{th:OCI_Kahan_optimal_Rz}
		Under Assumption~\ref{ass:J}, if $\mathbf{R} = \zeros$, the pair $(\mathbf{K}^\star,\mathbf{B}^\star)$ is a Kahan-family-optimal solution to the OCI problem \eqref{eq:OCI_orig_prob}, where $	(\mathbf{B}^\star, \boldsymbol{\omega}^\star)\in$
		\begin{equation}\label{eq:OCI_Kahan_optimal_Rz}
				\begin{aligned}
					 & \!\!\!\!\underset{\substack{\mathbf{B}\in S_+^n,\boldsymbol{\omega} \in \Delta^M}}{\argmin} &\quad \!\!\!\!\!&\!\!\! J(\mathbf{B})\\
						&\!\!\!\!\quad\mathrm{s.t.} && \!\!\!\!\!\begin{bmatrix} \mathbf{B}&  \mathbf{I} \\ \mathbf{I} & \mathbf{H}^\top \mathbf{C}^{-\!\top}\!\left( \sum\nolimits_{b=1}^M\!\boldsymbol{\omega}_b\mathbf{Y}_b\right)\! \mathbf{C}^{-1} \mathbf{H}\end{bmatrix} \! \succeq \zeros\!\!
					\end{aligned}
			\end{equation}
		and \vspace{-0.3cm}
		\begin{equation*}
				\begin{split}
						\mathbf{K}^\star \! := \!&\left(\mathbf{H}^\top \mathbf{C}^{-\top}\left( \sum\nolimits_{b=1}^M\!\boldsymbol{\omega}^\star_b\mathbf{Y}_b\right) \mathbf{C}^{-1} \mathbf{H}\right)^{\!\!\!-\!1}\\
						& \quad\!\mathbf{H}^\top \mathbf{C}^{-\top}\left( \sum\nolimits_{b=1}^M\!\boldsymbol{\omega}^\star_b\mathbf{Y}_b\right) \mathbf{C}^{-1}.
					\end{split}
			\end{equation*}
		Furthermore, \eqref{eq:OCI_Kahan_optimal_Rz} is feasible if and only if $\mathbf{H}^\top \mathbf{C}^{-\top}\mathbf{W}^\top\mathbf{W}\mathbf{C}^{-1} \mathbf{H}$ is full rank or, equivalently, the OCI problem \eqref{eq:OCI_orig_prob} is feasible, where $\mathbf{W} := [\mathbf{W}_1^\top \; \mathbf{W}_2^\top \; \cdots \; \mathbf{W}_M^\top]^\top$.
	\end{theorem}
	\begin{proof}
		See Appendix~\ref{sec:proof_th_OCI_Kahan_optimal_Rz}.
	\end{proof}

\begin{theorem}[$\!${\cite[Corollary~2]{PedrosoBatistaEtAl2026OCI}}]\label{th:OCI_Kahan_optimal_Rpd}
	Under Assumption~\ref{ass:J}, if $\mathbf{R}\succ \zeros$, the pair $(\mathbf{K}^\star,\mathbf{B}^\star)$ is a Kahan-family-optimal solution to the OCI problem \eqref{eq:OCI_orig_prob}, where $(\mathbf{U}^\star, \mathbf{B}^\star, \boldsymbol{\omega}^\star)\in$
	\begin{equation}\label{eq:OCI_Kahan_optimal_Rdp}
		 \begin{aligned}
						& \!\!\! \underset{\substack{\mathbf{U},\mathbf{B}\in S_+^n,\\\boldsymbol{\omega} \in \Delta^M}}{\argmin} &\quad \!\!\!\!\!&J(\mathbf{B})\\
					\!\!\!& \!\!\!\quad\mathrm{s.t.} && \!\!\begin{bmatrix} \mathbf{B}&  \mathbf{I} \\ \mathbf{I} & \mathbf{H}^\top\mathbf{R}^{-1}\mathbf{H}-\mathbf{U}\end{bmatrix}  \succeq \zeros\\
					&  && \!\!\begin{bmatrix} \mathbf{U}&  \mathbf{H}^\top\mathbf{R}^{-1}\mathbf{C} \\ (\mathbf{H}^\top\mathbf{R}^{-1}\mathbf{C} )^\top & \!\!\sum_{b=1}^M\!\boldsymbol{\omega}_b\mathbf{Y}_b +\mathbf{C}^\top\! \mathbf{R}^{-1}\mathbf{C}\end{bmatrix} \! \succeq \!\zeros\!\!\!\!  %
				\end{aligned} \!\!\!
		\end{equation}
	and \vspace{-0.3cm}
	\begin{equation*}
		\begin{split}
				\mathbf{K}^\star \! := \!&\left(\!\!\mathbf{H}^\top\!\mathbf{R}^{-\!1}\!\!\left(\!\!\mathbf{R}\!-\!\mathbf{C}\!\left( \sum_{b=1}^M\!\boldsymbol{\omega}^\star_b\mathbf{Y}_b\!+\!\mathbf{C}^\top\! \mathbf{R}^{-\!1}\mathbf{C}\!\!\right)^{\!\!+}\!\!\! \!\mathbf{C}^\top\!\!\right)\!\!\mathbf{R}^{\!-\!1}\mathbf{H}\!\!\right)^{\!\!\!-\!1}\\
				&\;\;\;\!\mathbf{H}^\top\!\mathbf{R}^{-\!1}\!\!\left(\!\!\mathbf{R}\!-\!\mathbf{C}\!\left( \sum_{b=1}^M\!\boldsymbol{\omega}^\star_b\mathbf{Y}_b\!+\!\mathbf{C}^\top\! \mathbf{R}^{-\!1}\mathbf{C}\!\!\right)^{\!\!+}\!\!\! \!\mathbf{C}^\top\!\!\right)\!\!\mathbf{R}^{\!-\!1}.
			\end{split}
	\end{equation*}
		Furthermore, \eqref{eq:OCI_Kahan_optimal_Rdp} is feasible if and only if $\mathbf{H}^\top\mathbf{R}^{-1} ( \mathbf{R}-\mathbf{C}(\mathbf{W}^\top\mathbf{W}+\mathbf{C}^\top \mathbf{R}^{-1}\mathbf{C})^+ \mathbf{C}^\top)\mathbf{R}^{-1}\mathbf{H}$ is full rank or, equivalently, the OCI problem \eqref{eq:OCI_orig_prob} is feasible, where $\mathbf{W} := [\mathbf{W}_1^\top \; \mathbf{W}_2^\top \; \cdots \; \mathbf{W}_M^\top]^\top$.
	\end{theorem}

	Notice that the feasible sets of \eqref{eq:OCI_Kahan_optimal_Rz} and \eqref{eq:OCI_Kahan_optimal_Rdp} are convex, since they are characterized by linear matrix inequalities. Therefore, if $J$ is convex, then  \eqref{eq:OCI_Kahan_optimal_Rz} and \eqref{eq:OCI_Kahan_optimal_Rdp} are convex optimization problems, which can be efficiently solved with global optimality guarantees and are robust to changes in input parameters~\cite{BoydVandenberghe2004}. Crucially, for the typical choices of $J$ such as the trace or determinant, the optimization problems  \eqref{eq:OCI_Kahan_optimal_Rz} and \eqref{eq:OCI_Kahan_optimal_Rdp} can be written as an SDP, for which well-performing off-the-shelf solvers with polynomial worst-case complexity exist \cite{VandenbergheBoyd1996}\cite[Section~6]{MOSEK2024}. 
	
	\begin{remark}
		To use the determinant as the criterion $J$, the optimization problems need to be slightly modified to be cast as an SDP. Since $\det(\mathbf{X}^{-1}) = \det(\mathbf{X})^{-1}$ and the logarithm is strictly increasing, one should modify the objective functions of \eqref{eq:OCI_Kahan_optimal_Rdp} and \eqref{eq:OCI_Kahan_optimal_Rz} to $-\mathrm{logdet}(\mathbf{H}^\top\mathbf{R}^{-1}\mathbf{H}-\mathbf{U})$ and $-\mathrm{logdet}(\mathbf{H}^\top \mathbf{C}^{-\!\top}\!( \sum_{b=1}^M\!\boldsymbol{\omega}_b\mathbf{Y}_b) \mathbf{C}^{-1} \mathbf{H})$, respectively. See \cite[Section~6.2.3]{MOSEK2024} for details on this technique.
	\end{remark}

\section{Covariance Intersection}\label{sec:CI}

In this section, we consider a general framework of the standard CI problem introduced in \cite{Julier1997,JulierUhlmann2017}. The goal is to estimate a state  $\mathbf{x} \in \R^n$. For that purpose, there are $N$ pieces of partial information characterized by $\mathbf{z}_i = \mathbf{H}_i\mathbf{x}+\mathbf{e}_i$ with $i = 1,2,\ldots,N$, where $\mathbf{H}_i \in \R^{o_i\times n}$ is known and $\mathbf{e}_i$ is zero-mean random noise. Similarly to the OCI problem in Section~\ref{sec:oci_probelm}, one can define $\mathbf{z} := [\mathbf{z}_1^\top,\ldots,\mathbf{z}_N^\top]^\top \in \R^o$, $\mathbf{H} := [\mathbf{H}_1^\top \;\; \cdots \;\; \mathbf{H}_N^\top]^\top \in \R^{o\times n}$, and $\mathbf{e} := [\mathbf{e}_1^\top,\ldots,\mathbf{e}_N^\top]^\top \in \R^o$ to characterize the information of the partial estimates as $\mathbf{z} = \mathbf{H}\mathbf{x} + \mathbf{e}$. Moreover,  we know a bound $\mathbf{X}_i \succ \zeros$ on the autocorrelations, i.e., $\EV[\mathbf{e}_i\mathbf{e}_i^\top] \preceq \mathbf{X}_i$ for all $i = 1,2,\ldots,N$, but the correlations $\EV[\mathbf{e}_i\mathbf{e}_j^\top]$ are unknown for all  $i\neq j$. The goal is to design a fusion gain $\mathbf{K}\in \R^{n\times o}$ for an unbiased estimator $\hat{\mathbf{x}} = \mathbf{K}\mathbf{z}$ such that a criterion $J(\mathbf{B})$ of a bound on the fused covariance $\mathbf{B} \succeq \mathbf{K} \EV[\mathbf{e}\mathbf{e}^\top]\mathbf{K}^\top$ is minimized, where $J$ satisfies Assumption~\ref{ass:J}.

This problem can be cast in the form of an OCI problem, described in Section~\ref{sec:oci}. Specifically, $\EV[\mathbf{e}\mathbf{e}^\top]$ is not fully known, but it has the form $\EV[\mathbf{e}\mathbf{e}^\top] = \mathbf{P}$, where $\mathbf{P}$ is in the set of admissible matrices
\begin{equation*}
	\Pcal = \{\mathbf{P}\in S_{++}^o : \mathbf{W}_i\mathbf{P}\mathbf{W}_i^\top \preceq \mathbf{X}_i \; \forall i \in \{1,2,\ldots, N\}\},
\end{equation*} %
where $\mathbf{W}_i = [\zeros_{o_i\times o_1} \,  \cdots \,  \zeros_{o_i\times o_{i-1}}\;  \mathbf{I}_{o_i}  \;  \zeros_{o_i\times o_{i+1}}\cdots \zeros_{o_i\times o_N}]$ for $i = 1,2,\ldots,N$. Notice that this is the same form as \eqref{eq:E_ee} with $\mathbf{R} = \zeros$ and $\mathbf{C}= \mathbf{I}$. Therefore, applying Theorem~\ref{th:OCI_Kahan_optimal_Rz} immediately allows us to characterize the Kahan-family-optimal solution to the CI problem resorting to the following result.

\begin{theorem}\label{th:CI}
	Under Assumption~\ref{ass:J}, the Kahan-family-optimal solution to the CI problem is $(\mathbf{K}^\star,\mathbf{B}^\star)$, where $(\mathbf{B}^\star, \boldsymbol{\omega}^\star)\in$
	\begin{equation}\label{eq:CI_sdp}
		\begin{aligned}
			 \;\;& \underset{\substack{\mathbf{B}\in S_+^n,\boldsymbol{\omega} \in \Delta^N}}{\argmin} &\quad \!\!\!\!\!&J(\mathbf{B})\\
			&\quad\mathrm{s.t.} && \!\!\begin{bmatrix} \mathbf{B}&  \mathbf{I} \\ \mathbf{I} & \sum_{i = 1}^{N} \boldsymbol{\omega}_i \mathbf{H}_i^\top \mathbf{X}_i^{-1}\mathbf{H}_i \end{bmatrix}  \succeq \zeros
		\end{aligned}
	\end{equation}
	and \vspace{-0.3cm}
	\begin{equation*}
		\begin{split} %
			\mathbf{K}^\star \! := \!&\left(\sum_{i = 1}^{N} \!\boldsymbol{\omega}_i^\star \mathbf{H}_i^\top \mathbf{X}_i^{-1}\mathbf{H}_i\!\right)^{\!\!\!-\!1}\!\!\!\!\left[\boldsymbol{\omega}_1^\star \mathbf{H}_1^\top \mathbf{X}_1^{-1}\, \cdots \; \boldsymbol{\omega}_N^\star \mathbf{H}_N^\top \mathbf{X}_N^{-1} \right]\!.
		\end{split}
	\end{equation*}
	Furthermore, both the CI problem and \eqref{eq:CI_sdp} are feasible if and only if $\mathbf{H}$ is full column rank.
\end{theorem}

\begin{mdframed}[style=callout]
	If common criteria for $J$ such as the trace or determinant are used, Theorem~\ref{th:CI} characterizes the Kahan-family-optimal solution of the CI problem as the solution to an SDP, which can be efficiently solved with polynomial worst-case complexity with off-the-shelf solvers.
\end{mdframed}

Interestingly, the Kahan family of bounding ellipsoids described in Section~\ref{sec:oci_sol} is a generalization of the family commonly used in the CI literature \cite{Julier1997,JulierUhlmann2017}. Therefore, if Theorem~\ref{th:CI} is particularized to a case of multiple full state estimates, the exact same filters as \cite{JulierUhlmann2017} will be obtained.

\begin{remark}
	Other CI-like problems can be easily cast in the OCI framework using a procedure very similar to that of this section. An example is fusion with unequal state vectors \cite{SijsHanebeckEtAl2013,NoackSijsEtAl2014}.
\end{remark}

\section{Split Covariance Intersection}\label{sec:SCI}

The SCI method \cite{LiNashashibiEtAl2013} was recently generalized in \cite{CrosAmblardEtAl2025} to deal with known correlated components instead of known uncorrelated components. In this section, we show that the unifying framework proposed in this paper allows us to efficiently treat the generalized SCI problem.

The goal is to estimate a state  $\mathbf{x} \in \R^n$. For that purpose, there are $N$ pieces of partial information characterized by $\mathbf{z}_i = \mathbf{H}_i\mathbf{x}+\mathbf{e}_i \in \R^{o_i}$ with $i = 1,2,\ldots,N$, where $\mathbf{H}_i \in \R^{o_i\times n}$ is known. In this case, we know that $\mathbf{e}_i$ is zero-mean random noise that can be split into two zero-mean components $\mathbf{e}_i^{(1)}$ and $\mathbf{e}_i^{(2)}$ as $\mathbf{e}_i = \mathbf{e}_i^{(1)} + \mathbf{e}_i^{(2)}$.  One can define $\mathbf{z} := [\mathbf{z}_1^\top,\ldots,\mathbf{z}_N^\top]^\top \in \R^o$, $\mathbf{H} := [\mathbf{H}_1^\top \;\; \cdots \;\; \mathbf{H}_N^\top]^\top \in \R^{o\times n}$, $\mathbf{e}^{(1)} := [\mathbf{e}_1^{(1)\top},\ldots,\mathbf{e}_N^{(1)\top}]^\top \in \R^o$, $\mathbf{e}^{(2)} := [\mathbf{e}_1^{(2)\top},\ldots,\mathbf{e}_N^{(2)\top}]^\top \in \R^o$, and $\mathbf{e} :=  \mathbf{e}^{(1)}+ \mathbf{e}^{(2)}$ to characterize the information of the partial estimates as $\mathbf{z} = \mathbf{H}\mathbf{x} + \mathbf{e}^{(1)}+ \mathbf{e}^{(2)}$. We know a bound $\mathbf{X}^{(1)}_i \succ \zeros$ on the autocorrelation of the first component, i.e., $\EV[\mathbf{e}_i^{(1)}\mathbf{e}_i^{(1)\top}] \preceq \mathbf{X}_i^{(1)}$ for all $i = 1,2,\ldots,N$. Additionally, the second components have known correlations, i.e., $ \EV[\mathbf{e}^{(2)}\mathbf{e}^{(2)\top}] = \mathbf{X}^{(2)} \succ \zeros$ and $\EV[\mathbf{e}^{(1)}\mathbf{e}^{(2)\top}] = \zeros$.\footnote{We can also assume both components are correlated with a known correlation $\EV[\mathbf{e}^{(1)}\mathbf{e}^{(2)\top}] = \mathbf{X}^{(1,2)}$. Nevertheless, after a transformation, this case can be reduced to $\EV[\mathbf{e}^{(1)}\mathbf{e}^{(2)\top}] = \zeros$. See \cite[Section~III-B]{CrosAmblardEtAl2025} for details.} Notice that the correlations $\EV[\mathbf{e}_i^{(1)}\mathbf{e}_j^{(1)\top}]$ are unknown for all  $i\neq j$. The standard SCI problem in \cite{LiNashashibiEtAl2013} corresponds to the case when $\mathbf{X}^{(2)}$ is block diagonal. 

We are interested in designing a gain $\mathbf{K}\in \R^{n\times o}$ for an unbiased estimator $\hat{\mathbf{x}} = \mathbf{K}\mathbf{z}$. The estimation error is $\mathbf{\tilde{x}} := \mathbf{\hat{x}}-\mathbf{x}$. Since the filter is unbiased it follows that $\mathbf{K}\mathbf{H} = \mathbf{I}$. Thus $\mathbf{\tilde{x}} = \mathbf{K}\mathbf{e}^{(1)} + \mathbf{K}\mathbf{e}^{(2)}$, which can be written in two components as $\mathbf{\tilde{x}} = \tilde{\mathbf{e}}^{(1)} + \tilde{\mathbf{e}}^{(2)}$, where $\tilde{\mathbf{e}}^{(1)} = \mathbf{K}\mathbf{e}^{(1)}$ and  $\tilde{\mathbf{e}}^{(2)} = \mathbf{K}\mathbf{e}^{(2)}$. The goal is to design $\mathbf{K}$ such that a criterion $J(\mathbf{B})$ of a bound on the fused covariance $\mathbf{B} \succeq  \mathbf{K} \EV[\mathbf{e}\mathbf{e}^\top]\mathbf{K}^\top$ is minimized, where $J$ satisfies Assumption~\ref{ass:J}. Moreover, we are interested in characterizing a bound on the first component of the estimation error $\mathbf{B}^{(1)} \succeq \EV[\tilde{\mathbf{e}}^{(1)}\tilde{\mathbf{e}}^{(1)\top}]$ and characterizing the second component with  $\mathbf{B}^{(2)} =  \EV[\tilde{\mathbf{e}}^{(2)}\tilde{\mathbf{e}}^{(2)\top}]$. Note that it holds that $\mathbf{B} = \mathbf{B}^{(1)} + \mathbf{B}^{(2)}$ and $ \EV[\tilde{\mathbf{e}}^{(1)}\tilde{\mathbf{e}}^{(2)\top}]   = \zeros$.

This problem can be cast in the form of an OCI problem, described in Section~\ref{sec:oci}. Specifically, $\EV[\mathbf{e}\mathbf{e}^\top]$ is not fully known but has the form $\EV[\mathbf{e}\mathbf{e}^\top] = \mathbf{X}^{(2)}+ \mathbf{P}$, where $\mathbf{P}$ is in the set of admissible matrices
\begin{equation*}
	\Pcal = \{\mathbf{P}\in S_{++}^o : \mathbf{W}_i\mathbf{P}\mathbf{W}_i^\top \preceq \mathbf{X}^{(1)}_i \; \forall i \in \{1,2,\ldots, N\}\},
\end{equation*}
where $\mathbf{W}_i = [\zeros_{o_i\times o_1}  \, \cdots \,  \zeros_{o_i\times o_{i-1}}\;  \mathbf{I}_{o_i}  \;  \zeros_{o_i\times o_{i+1}}\cdots \zeros_{o_i\times o_N}]$ for $i = 1,2,\ldots,N$. Notice that this is the same form as \eqref{eq:E_ee} with $\mathbf{R} = \mathbf{X}^{(2)}$ and $\mathbf{C}= \mathbf{I}$.  Therefore, applying Theorem~\ref{th:OCI_Kahan_optimal_Rpd} allows to characterize the Kahan-family-optimal solution of the generalized SCI problem resorting to the following result.

\begin{theorem}\label{th:SCI}
	Under Assumption~\ref{ass:J}, the Kahan-family-optimal solution to the generalized SCI problem is $(\mathbf{K}^\star,\mathbf{B}^{\star(1)}, \mathbf{B}^{\star(2)})$, where $(\mathbf{U}^\star, \mathbf{B}^\star, \boldsymbol{\omega}^\star)\in$
	\begin{equation}\label{eq:SCI_sdp}
	\begin{aligned}
		 & \!\!\!\!\underset{\substack{\mathbf{U},\mathbf{B}\in S_+^n\\\boldsymbol{\omega} \in \Delta^N}}{\argmin} \!\!\!& \!\!\!\!\!\!\!&J(\mathbf{B})\\
		&\!\;\;\mathrm{s.t.} && \!\begin{bmatrix} \mathbf{B}&  \mathbf{I} \\ \mathbf{I} & \mathbf{H}^\top{\mathbf{X}^{(2)}}^{-1}\mathbf{H}-\mathbf{U}\end{bmatrix}  \succeq \zeros\\
		&  && \!\!\begin{bmatrix} \mathbf{U}&  \mathbf{H}^\top{\mathbf{X}^{(2)}}^{-1} \\ {\mathbf{X}^{(2)}}^{\!-1}\mathbf{H}\;\; & \sum\limits_{i=1}^N\!\boldsymbol{\omega}_i\mathbf{W}_i^\top{\mathbf{X}^{(1)}_i}^{\!-1}\!\!\mathbf{W}_i +\! {\mathbf{X}^{(2)}}^{-1}\end{bmatrix} \!\! \succeq \!\zeros,\!\!\!\!  %
	\end{aligned}\!\!\!\!\!\!\!
\end{equation}
	\begin{equation*}
		\begin{split}
		 &\mathbf{B}^{\star(1)}  =  \mathbf{B} - \mathbf{K}^\star\mathbf{X}^{(2)} \mathbf{K}^{\star\top},\\ 	&\mathbf{B}^{\star(2)} =   \mathbf{K}^\star \mathbf{X}^{(2)} \mathbf{K}^{\star \top},\\
			&\mathbf{K}^\star  =  \left(\!\!\mathbf{H}^\top\!{\mathbf{X}^{(2)}}^{-1} \!\!\left(\!\!\mathbf{X}^{(2)} \!-\!\left( \mathbf{Y}^\star\!+\!{\mathbf{X}^{(2)}}^{-1} \right)^{\!\!\!-\!1}\right)\!\!{\mathbf{X}^{(2)}}^{-1} \mathbf{H}\right)^{\!\!\!-\!1}\\
			&\!\qquad\mathbf{H}^\top\!{\mathbf{X}^{(2)}}^{-1} \!\!\left(\!\!\mathbf{X}^{(2)} \!-\!\left( \mathbf{Y}^\star\!+\!{\mathbf{X}^{(2)}}^{-1} \right)^{\!\!\!-\!1}\right)\!\!{\mathbf{X}^{(2)}}^{-1},
\end{split} 
	\end{equation*}
and $\mathbf{Y}^\star  := \sum_{i=1}^N\!\boldsymbol{\omega}^\star_i\mathbf{W}_i^\top{\mathbf{X}^{(1)}_i}^{-1}\mathbf{W}_i.$
Furthermore, both the SCI problem and \eqref{eq:SCI_sdp} are feasible if and only if $\mathbf{H}$ is full column rank.
\end{theorem}

Since the Kahan family of bounding ellipsoids described in Section~\ref{sec:oci_sol} is a generalization of the family commonly used in the SCI literature, if Theorem~\ref{th:SCI} is particularized to the standard SCI problem with two full-state estimates, the exact same filters of \cite{LiNashashibiEtAl2013} will be obtained. 

\begin{mdframed}[style=callout]
	If common criteria for $J$ such as the trace or determinant are used, Theorem~\ref{th:SCI} characterizes the Kahan-family-optimal solution to the generalized and standard SCI problems as an SDP, which can be efficiently solved with polynomial worst-case complexity with off-the-shelf solvers.
\end{mdframed}

An implementation of the solution to the CI and SCI problems using an SDP, together with a comparison to the corresponding state-of-the-art methods, is available in an open-source repository at \href{https://github.com/decenter2021/unification-CI}{github.com/decenter2021/unification-CI}.

\section{Conclusion}

In this paper, we showed that the overlapping covariance intersection (OCI) framework provides a unifying formulation for several covariance intersection–based fusion problems. The proposed OCI framework admits a family-optimal solution that can be computed as the solution of a semidefinite program (SDP), enabling efficient numerical implementation. This unifying approach allows family-optimal solutions for multiple CI variants to be obtained within a single optimization framework. In particular, we demonstrated that the family-optimal solutions to both the standard covariance intersection (CI) problem and the split covariance intersection (SCI) problem, for a general number of partial estimates, can be characterized as solutions to SDPs. When specialized to their respective settings, the resulting solutions recover the previously reported state-of-the-art family-optimal solutions for CI and SCI. Overall, the results presented in this paper enable the real-time deployment of CI-based fusion methods in ultra large-scale estimation applications, such as cooperative localization.

\section*{Acknowledgment}
The work of L. Pedroso and P. Batista was supported by the LARSyS FCT funding (UID/50009/2025: {\small DOI \href{https://doi.org/10.54499/UID/50009/2025}{\texttt{10.54499/UID/50009/2025}}}; LA/P/0083/2020: {\small DOI \href{https://doi.org/10.54499/LA/P/0083/2020}{\texttt{10.54499/LA/P/0083/2020}}}).

\vspace{-0.1cm}
\appendix
\vspace{-0.1cm}
\subsection{Proof of Lemma~\ref{lem:kahan_conservative}}\label{sec:proof_lem_kahan_conservative}

The proof consists of showing that $\mathbf{P} \in \Pcal$ implies $\mathbf{P}\in  \Pcal_{\mathrm{KF}}(\boldsymbol{\omega})$ for all $\boldsymbol{\omega} \in  \Delta^M$. Consider any  $\mathbf{P} \in \Pcal$. From the definition of $\Pcal$ in \eqref{eq:Pcal_def} and \cite[Lemma~1]{PedrosoBatistaEtAl2026OCI} it follows that $\mathbf{P}^{-1} \succeq \mathbf{Y}_b$ for all $b = 1,2,\ldots, M$. Thus, for any $\boldsymbol{\omega} \in  \Delta^M$, one can write $\mathbf{P}^{-1} = \sum_{b = 1}^M \boldsymbol{\omega}_b \mathbf{P}^{-1} \succeq \sum_{b = 1}^M \boldsymbol{\omega}_b \mathbf{Y}_b$. Therefore, $\mathbf{P} \in  	\textstyle\Pcal_{\mathrm{KF}}(\boldsymbol{\omega})$ according to the definition of $\Pcal_{\mathrm{KF}}(\boldsymbol{\omega})$ in \eqref{eq:kahan_family}, which concludes the proof.

\vspace{-0.1cm}
\subsection{Proof of Theorem~\ref{th:OCI_Kahan_optimal_Rz}}\label{sec:proof_th_OCI_Kahan_optimal_Rz}

The following results on the positive (semi)definiteness of block matrices will be instrumental in establishing the results in this paper.
\vspace{-0.1cm}
\begin{proposition}[$\!\!${\cite[Prop.~I.1]{PedrosoBatistaEtAl2026OCI}}]\label{prop:schur}
	$\!$Let {\small $\mathbf{X} \!\!=\!\! \begin{bmatrix} \!\mathbf{A} &\! \!\!\!\mathbf{B} \\  \!\mathbf{B}^\top & \!\!\!\!\mathbf{C}
	\end{bmatrix}$} be a symmetric block matrix with $\mathbf{A}\in \R^{p\times p}$ and $\mathbf{C}\in \R^{q\times q}$. Then:
	\begin{enumerate}[(i)]
		\item If $\mathbf{B} \!=\! \mathbf{I}_p$, then $\mathbf{X}\!\succeq \!\zeros \! \implies\! {\mathbf{A} \!\succ\! \zeros},\; {\mathbf{C} \!\succ\! \zeros},\;\mathbf{A} \!\succeq \! \mathbf{C}^{-1},\; \mathbf{A}^{-1} \preceq \mathbf{C}$;\label{schur:BIr}
		\item  If $\mathbf{B} = \mathbf{I}_p$, then $\mathbf{A} \succ \zeros, \mathbf{C} \succ \zeros$ and either $ \mathbf{A} \succeq \mathbf{C}^{-1}$ or $\mathbf{A}^{-1} \preceq \mathbf{C}$ implies $\mathbf{X}\succeq \zeros$;\label{schur:BIl}
		\item If $\mathbf{X} \!\succ \!\zeros$ and $\mathbf{Q}\!\in\! S_{+}^q$, then {\small $\mathbf{X}^{-1} \!\succeq \!\begin{bmatrix}\! \mathbf{Q} & \!\!\!\zeros \\ \! \zeros &\!\!\!\zeros\end{bmatrix}$} $\!\!\implies \!\!\mathbf{A}^{-1} \!\succeq\! \mathbf{Q}$.\label{schur:ineq11}
	\end{enumerate}
\end{proposition}
\vspace{-0.1cm}

The first step of the proof is decoupling the optimization of the gain and covariance bounds in problem~\eqref{eq:OCI_orig_prob}. Indeed, the following result establishes an equivalence between the solutions of \eqref{eq:OCI_orig_prob} with $\mathbf{R} = \zeros$ and of
\begin{subequations}\label{eq:OCI_Y}
	\begin{alignat}{2}
		& \!\!\! \underset{\substack{\mathbf{B}\in S_+^n ,\mathbf{Y}\in S_+^o}}{\argmin} &\quad \!\!\!\!\!&J(\mathbf{B})\\
		\!\!\!& \!\!\!\quad\mathrm{s.t.} && \!\!\begin{bmatrix} \mathbf{B}&  \mathbf{I} \\ \mathbf{I} & \mathbf{H}^\top \mathbf{C}^{-\!\top}\mathbf{Y}\mathbf{C}^{-1}\mathbf{H}\end{bmatrix}  \succeq \zeros \label{eq:OCI_Y_B} \\
		&&&	\mathbf{Y} \preceq \mathbf{P}^{-1},\; \forall \mathbf{P} \in \Pcal. \label{eq:OCI_Y_Y} 
	\end{alignat}
\end{subequations}

\begin{lemma}\label{lem:equiv_OCI_prob_1}
	Assume that Assumption~\ref{ass:J} holds. If $(\mathbf{Y}^\star,\mathbf{B}^\star)$ is a solution to \eqref{eq:OCI_Y}, then the pair $(\mathbf{K}^\star,\mathbf{B}^\star)$ is a solution to \eqref{eq:OCI_orig_prob}, with $\mathbf{K}^\star :=  (\mathbf{H}^\top \mathbf{C}^{-\!\top} \mathbf{Y}^\star \mathbf{C}^{-1}\mathbf{H})^{-1}  \mathbf{H}^\top \mathbf{C}^{-\!\top} \mathbf{Y}^\star \mathbf{C}^{-1}$.  If $(\mathbf{K}^{\circ},\mathbf{B}^\circ)$ is a solution to \eqref{eq:OCI_orig_prob}, then $(\mathbf{Y}^\bullet,\mathbf{B}^\bullet)$ is a solution to \eqref{eq:OCI_Y}, with $\mathbf{Y}^\bullet := (\mathbf{K}^\circ\mathbf{C})^\top\mathbf{B}^{\circ+}\mathbf{K}^\circ\mathbf{C}$ and $\mathbf{B}^\bullet  :=	 (\mathbf{H}^\top \mathbf{C}^{-\!\top} \mathbf{Y}^\bullet \mathbf{C}^{-1}\mathbf{H})^{-1}$.
\end{lemma} \vspace{-0.3cm}%
\begin{proof}
	We start by showing that if $(\mathbf{Y}^\star,\mathbf{B}^\star)$ is in the feasible domain of \eqref{eq:OCI_Y}, then  $(\mathbf{K}^\star,\mathbf{B}^\star)$ is in the feasible domain of \eqref{eq:OCI_orig_prob}. First, notice that $\mathbf{K}^\star$ is well-defined because applying Proposition~\ref{prop:schur}\eqref{schur:BIr} to \eqref{eq:OCI_Y_B} yields $\mathbf{H}^\top \mathbf{C}^{-\!\top}\mathbf{Y}^\star\mathbf{C}^{-1}\mathbf{H} \succ \zeros$. Then, it follows that $\mathbf{K}^\star\mathbf{H} = \mathbf{I}$. Second, by hypothesis, $\mathbf{Y}^\star \preceq \mathbf{P}^{-1}$ for all $\mathbf{P}\in \Pcal$, so $\mathbf{H}^\top \mathbf{C}^{-\!\top} \mathbf{Y}^\star \mathbf{C}^{-1}\mathbf{H} \preceq  \mathbf{H}^\top \mathbf{C}^{-\!\top} \mathbf{P}^{-1} \mathbf{C}^{-1}\mathbf{H}$. Therefore, by Proposition~\ref{prop:schur}\eqref{schur:BIr} and \eqref{eq:OCI_Y} it follows that 
	\begin{equation}\label{eq:star_ineq1}
		\!\!\!\mathbf{B}^\star \! \succeq\! (\mathbf{H}^\top\! \mathbf{C}^{-\!\top}\! \mathbf{Y}^\star \mathbf{C}^{-1}\mathbf{H})^{-1} \!\succeq\! (\mathbf{H}^\top \!\mathbf{C}^{-\!\top}\!\mathbf{P}^{-1} \mathbf{C}^{-1}\mathbf{H})^{-1}\vspace{-0.1cm}
	\end{equation}
	for all $\mathbf{P}\in \Pcal$. Matrix $\mathbf{Y}^\star$ is real, symmetric, and positive semidefinite, so by the spectral theorem for real symmetric matrices it admits a factorization $\mathbf{Y}^\star = [\mathbf{V}\; \mathbf{V}_\perp]\diag(\mathbf{D},\zeros)[\mathbf{V}\; \mathbf{V}_\perp]^\top,$
	where $\mathbf{D}\in S_{++}^r$ is a diagonal matrix with $r= \rank(\mathbf{Y}^\star)$ and the columns of $[\mathbf{V}\; \mathbf{V}_\perp]$ form an orthonormal basis for $\R^o$. Since, by hypothesis, $\mathbf{Y}^\star \preceq \mathbf{P}^{-1}$ for all $\mathbf{P}\in \Pcal$, using the factorization for $\mathbf{Y}^\star$, one can write $	\diag(\mathbf{D},\zeros) \preceq	[\mathbf{V} \; \mathbf{V}_\perp]^\top\!\mathbf{P}^{-1}[\mathbf{V} \; \mathbf{V}_\perp] = 	([\mathbf{V} \; \mathbf{V}_\perp]^\top\!\mathbf{P}[\mathbf{V} \; \mathbf{V}_\perp])^{-1}\!\!$.
	Thus, by Proposition~\ref{prop:schur}\eqref{schur:ineq11}, $\mathbf{V}^\top\mathbf{P}\mathbf{V} \preceq \mathbf{D}^{-1}$ for all $\mathbf{P}\in \Pcal$. As a result, one can write
	\begin{equation}\label{eq:star_ineq2}
			\mathbf{Y}^\star\mathbf{P}\mathbf{Y}^\star = \mathbf{V}\mathbf{D}\mathbf{V}^\top \!\mathbf{P} \mathbf{V} \mathbf{D}\mathbf{V}^\top \!\!\preceq \mathbf{V} \mathbf{D}\mathbf{D}^{-1}\mathbf{D} \mathbf{V}^\top  \!\!=\! \mathbf{Y}^\star\!. \vspace{-0.1cm}
	\end{equation}
	Finally, after algebraic manipulation resorting to \eqref{eq:star_ineq1} and  \eqref{eq:star_ineq2}, it follows that $\mathbf{K}^\star \mathbf{C} \mathbf{P} \mathbf{C}^\top \mathbf{K}^{\star\top} \preceq \mathbf{B}^\star$ for all $\mathbf{P}\in \Pcal$. Therefore, $(\mathbf{K}^\star,\mathbf{B}^\star)$ is in the feasible domain of \eqref{eq:OCI_orig_prob}.
	
	Second, we show that if $(\mathbf{K}^{\circ},\mathbf{B}^\circ)$ is in the feasible domain of \eqref{eq:OCI_orig_prob}, then $(\mathbf{Y}^\bullet,\mathbf{B}^\bullet)$ is in the feasible domain of \eqref{eq:OCI_Y} and  $J(\mathbf{B}^\bullet) \leq J(\mathbf{B}^\circ)$. Consider a factorization $\mathbf{Y}^\bullet = [\mathbf{V}\; \mathbf{V}_\perp]\diag(\mathbf{D},\zeros)[\mathbf{V}\; \mathbf{V}_\perp]^\top$. Define $\mathbf{S}_\perp$ as a matrix whose columns form an orthonormal basis for $\mathbf{C}\mathbf{V}_\perp$ and $\mathbf{S}$ such that the columns of $[\mathbf{S}\; \mathbf{S}_\perp]$ form an orthonormal basis for $\R^o$. First, following exactly the same steps as in the proof of the second statement of \cite[Lemma~3]{PedrosoBatistaEtAl2026OCI} with $\mathbf{R} = \zeros$ we conclude that $\mathbf{P}^{-1}\succeq (\mathbf{K}^\circ\mathbf{C})^\top{\mathbf{B}}^{\circ+} \mathbf{K}^\circ\mathbf{C} = \mathbf{Y}^\bullet$ for all $\mathbf{P}\in \Pcal$, so constraint \eqref{eq:OCI_Y_Y} is satisfied. Moreover, we also conclude that $\mathbf{B}
	= (\mathbf{H}^\top\mathbf{S}\left(\mathbf{S}^\top \mathbf{C}\mathbf{Y}^{\bullet+}\mathbf{C}^\top\mathbf{S}\right)^{-1}\mathbf{S}^\top\mathbf{H})^{-1}$ is well-defined and satisfies $J(\mathbf{B})\leq J(\mathbf{B}^\circ)$.  Expanding $\mathbf{S}(\mathbf{S}^\top \mathbf{C}\mathbf{Y}^{\bullet+}\mathbf{C}^\top\mathbf{S})^{-1}\mathbf{S}^\top$ yields
	\begin{equation*}
		\begin{split}
				&	\mathbf{S}(\mathbf{S}^\top \mathbf{C}[\mathbf{V}\; \mathbf{V}_\perp]\diag(\mathbf{D}^{-1},\zeros)[\mathbf{V}\; \mathbf{V}_\perp]^\top\mathbf{C}^\top\mathbf{S})^{-1}\mathbf{S}^\top \\
				 = \:& \mathbf{S}(\mathbf{V}^\top\mathbf{C}^\top\mathbf{S})^{-1} \mathbf{D}(\mathbf{V}^\top\mathbf{C}^\top\mathbf{S})^{-\top}\mathbf{S}^\top\\
				 = \:& \mathbf{C}^{-\!\top}[\mathbf{V}\;\mathbf{V}_\perp] [\mathbf{V}\;\mathbf{V}_\perp]^\top\! \mathbf{C}^\top \!\mathbf{S}(\mathbf{V}^{\!\top\!}\mathbf{C}^\top\!\mathbf{S})^{-1} \mathbf{D}(\mathbf{V}^{\!\top\!}\mathbf{C}^\top\mathbf{S})^{-\!\top\!}\mathbf{S}^\top\\
				 = \:& \mathbf{C}^{-\!\top}[\mathbf{V}\;\mathbf{V}_\perp] [\mathbf{I}\; \zeros]^\top \mathbf{D}  [\mathbf{I}\; \zeros] [\mathbf{V}\;\mathbf{V}_\perp]^\top \mathbf{C}^{-1}\\
				 = \: & \mathbf{C}^{-\!\top}[\mathbf{V}\;\mathbf{V}_\perp]\diag(\mathbf{D},\zeros)[\mathbf{V}\;\mathbf{V}_\perp]^\top \mathbf{C}^{-1}\\
				 = \: & \mathbf{C}^{-\!\top}\mathbf{Y}^\bullet \mathbf{C}^{-1},
		\end{split}
	\end{equation*}%
	where we used the fact that $\mathbf{V}^\top\mathbf{C}^\top\mathbf{S}$ is a $r\times r$ invertible matrix, since $\rank(\mathbf{S}^\top\mathbf{C}\mathbf{V}) = \rank([\mathbf{S}^\top\!\mathbf{C}\mathbf{V} \;\! \zeros]) = \rank(\mathbf{S}^\top\mathbf{C}[\mathbf{V} \;\! \mathbf{V}_\perp]) = \rank(\mathbf{S}^\top\mathbf{C}) = \rank(\mathbf{S}) = o\!-\!\rank(\mathbf{S}_\perp) \!=\! o\!-\!\rank(\mathbf{C}\mathbf{V}_\perp) \!=\!  o\!-\!\rank(\mathbf{V}_\perp) \!=\! o\!-\!(o\!-\!r) \!=\! r$.
	It follows that $(\mathbf{H}^\top\mathbf{S}\left(\mathbf{S}^\top \mathbf{C}\mathbf{Y}^{\bullet+}\mathbf{C}^\top\mathbf{S}\right)^{-1}\mathbf{S}^\top\mathbf{H})^{-1} = (\mathbf{H}^\top \mathbf{C}^{-\!\top} \mathbf{Y}^\bullet \mathbf{C}^{-1}\mathbf{H})^{-1} = \mathbf{B}^\bullet$ and $J(\mathbf{B}^\bullet)\leq J(\mathbf{B}^\circ)$. As a result, since $\mathbf{B}^\bullet  = (\mathbf{H}^\top \mathbf{C}^{-\!\top} \mathbf{Y}^\bullet \mathbf{C}^{-1}\mathbf{H})^{-1} \succ \zeros$, by Proposition~\ref{prop:schur}\eqref{schur:BIl}, constraint \eqref{eq:OCI_Y_B} is satisfied. Therefore, $(\mathbf{Y}^\bullet,\mathbf{B}^\bullet)$ is in the feasible domain of \eqref{eq:OCI_Y}.

	Third, let $(\mathbf{Y}^\star,\mathbf{B}^\star)$ be a solution to \eqref{eq:OCI_Y}. Therefore, $(\mathbf{Y}^\star,\mathbf{B}^\star)$ is in the feasible domain of \eqref{eq:OCI_Y} and, by the previous analysis $(\mathbf{K}^\star,\mathbf{B}^\star)$ is in the feasible domain of \eqref{eq:OCI_orig_prob}. We show that $(\mathbf{K}^\star,\mathbf{B}^\star)$ is a solution to \eqref{eq:OCI_orig_prob} by contradiction. Assume, by contradiction, that there is $(\mathbf{K}^{\circ},\mathbf{B}^\circ)$ in the feasible domain of \eqref{eq:OCI_orig_prob} such that $J(\mathbf{B}^\circ) < J(\mathbf{B}^\star)$. By the previous analysis, it follows that the tuple $(\mathbf{Y}^\bullet, \mathbf{B}^\bullet)$ is in the feasible domain of \eqref{eq:OCI_Y} and $J(\mathbf{B}^\bullet) \leq J(\mathbf{B}^\circ)$. Since $(\mathbf{Y}^\star,\mathbf{B}^\star)$ is a solution to \eqref{eq:OCI_Y}, then $J(\mathbf{B}^\star) \leq J(\mathbf{B}^\bullet) \leq J(\mathbf{B}^\circ)$. Bringing everything together yields $J(\mathbf{B}^\circ) < J(\mathbf{B}^\star)  \leq J(\mathbf{B}^\bullet) \leq J(\mathbf{B}^\circ)$, which is a contradiction.
	
	Fourth, let $(\mathbf{K}^{\circ},\mathbf{B}^\circ)$ be a solution to \eqref{eq:OCI_orig_prob}. Therefore, $(\mathbf{K}^{\circ},\mathbf{B}^\circ)$ is in the feasible domain of  \eqref{eq:OCI_orig_prob} and, by the previous analysis, the tuple $(\mathbf{Y}^\bullet, \mathbf{B}^\bullet)$ is in the feasible domain of \eqref{eq:OCI_Y} and $J(\mathbf{B}^\bullet) \leq J(\mathbf{B}^\circ)$. We show that  $(\mathbf{Y}^\bullet, \mathbf{B}^\bullet)$ is a solution to \eqref{eq:OCI_Y} by contradiction. Assume, by contradiction, that there is  $(\mathbf{Y}^\star,\mathbf{B}^\star)$ in the feasible domain of \eqref{eq:OCI_Y} such that $J(\mathbf{B}^\star) < J(\mathbf{B}^\bullet)$. By the previous analysis, it follows that $(\mathbf{K}^\star,\mathbf{B}^\star)$ is in the feasible domain of \eqref{eq:OCI_orig_prob}. Since $(\mathbf{K}^{\circ},\mathbf{B}^\circ)$ is a solution to \eqref{eq:OCI_orig_prob}, then $J(\mathbf{B}^\circ)\leq J(\mathbf{B}^\star)$. Bringing everything together yields $J(\mathbf{B}^\star) < J(\mathbf{B}^\bullet) \leq J(\mathbf{B}^\circ) \leq J(\mathbf{B}^\star)$, which is a contradiction.
\end{proof}

In the following lemma, we establish a necessary and sufficient condition for the feasibility of \eqref{eq:OCI_orig_prob}.

\begin{lemma}\label{lem:feas_iff}
	Under Assumption~\ref{ass:J}, the OCI problem \eqref{eq:OCI_orig_prob} is feasible if and only if $\mathbf{H}^\top\mathbf{C}^{-\top}\mathbf{W}^\top\mathbf{W}\mathbf{C}^{-1}\mathbf{H}$ is full rank.
\end{lemma}%
\begin{proof}
	Lemma~\ref{lem:equiv_OCI_prob_1} establishes that the OCI problem \eqref{eq:OCI_orig_prob} is feasible if and only if \eqref{eq:OCI_Y} is feasible. Thus, we only need to show the equivalence between the rank condition and  the feasibility of \eqref{eq:OCI_Y}. 
	In one direction of the equivalence, let $\mathbf{H}^\top\!\mathbf{C}^{-\!\top}\mathbf{W}^\top\!\mathbf{W}\mathbf{C}^{-1}\mathbf{H}$ be full rank. By  \cite[Prop.~I.2]{PedrosoBatistaEtAl2026OCI}, it follows that $\mathbf{P}^{-1} \succeq \mathbf{Y} := \mathbf{W}^\top\!\diag(\mathbf{X}_1, \ldots,\mathbf{X}_M)^{-1}\mathbf{W}/M$ for all $\mathbf{P}\!\in \!\Pcal$. Thus, $\mathbf{Y}$ satisfies \eqref{eq:OCI_Y_Y}. Notice that $\mathbf{Y}$ has the same column space as $\mathbf{W}^\top\mathbf{W}$. Therefore, by hypothesis, $\mathbf{H}^\top\mathbf{C}^{-\top}\mathbf{Y}\mathbf{C}^{-1}\mathbf{H}$ is also full rank. As a result, by Proposition~\ref{prop:schur}\eqref{schur:BIl} the pair $(\mathbf{B},\mathbf{Y})$ with $\mathbf{B} = (\mathbf{H}^\top\mathbf{C}^{-\top}\mathbf{Y}\mathbf{C}^{-1}\mathbf{H})^{-1}$ satisfies \eqref{eq:OCI_Y_B}. Thus, \eqref{eq:OCI_Y} is feasible.
	In the other direction, assume that \eqref{eq:OCI_Y} is feasible. Then, there exists $\mathbf{Y}$ such that $\mathbf{Y} \preceq \mathbf{P}^{-1}$ for all $\mathbf{P}\in \Pcal$ and, by Proposition~\ref{prop:schur}\eqref{schur:BIr}, $\mathbf{H}^\top\mathbf{C}^{-\top}\mathbf{Y}\mathbf{C}^{-1}\mathbf{H}$ is full rank. Since the bounded components of $\mathbf{P}$ are characterized by the row space of $\mathbf{W}$, it follows that $\col{\mathbf{Y}}\subseteq \col{\mathbf{W}^\top\mathbf{W}}$, therefore $\mathbf{H}^\top\mathbf{C}^{-\top}\mathbf{W}^\top\mathbf{W}\mathbf{C}^{-1}\mathbf{H}$ is also full rank, which concludes the proof.
\end{proof}

Finally, we prove the theorem resorting to Lemmas~\ref{lem:equiv_OCI_prob_1} and~\ref{lem:feas_iff}. First, one can decouple the optimization over $\boldsymbol{\omega}$ from the remaining decision variables in \eqref{eq:OCI_kahan}. For a fixed $\boldsymbol{\omega}$, \eqref{eq:OCI_kahan} is in the form of  \eqref{eq:OCI_orig_prob} and one can use Lemma~\ref{lem:equiv_OCI_prob_1} to characterize the solutions to \eqref{eq:OCI_kahan} for the fixed $\boldsymbol{\omega}$ using a problem of the form of \eqref{eq:OCI_Y}. Then, one can reintroduce the optimization over $\boldsymbol{\omega}$ to show the first statement of the theorem. Second, we turn to the second statement of the theorem. In one direction, if $(\mathbf{B},\boldsymbol{\omega})$ is feasible for \eqref{eq:OCI_kahan}, then $(\mathbf{B},\mathbf{Y})$ with  $\mathbf{Y} = \sum_{b=1}^M\boldsymbol{\omega}_b\mathbf{Y}_b$ is immediately feasible for \eqref{eq:OCI_Y}, and, by Lemmas~\ref{lem:equiv_OCI_prob_1} and~\ref{lem:feas_iff}, it follows that $\mathbf{H}^\top\mathbf{C}^{-\top}\mathbf{W}^\top\mathbf{W}\mathbf{C}^{-1}\mathbf{H}$ is full rank. In the other direction, let $\mathbf{H}^\top\mathbf{C}^{-\top}\mathbf{W}^\top\mathbf{W}\mathbf{C}^{-1}\mathbf{H}$ be full rank. Then for $\boldsymbol{\omega}_b = 1/M$ for all $b = 1,2,\ldots, M$, $\mathbf{Y} = \sum_{b=1}^M\boldsymbol{\omega}_b\mathbf{Y}_b$ has the same column space as $\mathbf{W}^\top\mathbf{W}$. Therefore, $\mathbf{H}^\top\mathbf{C}^{-\top}\mathbf{Y}\mathbf{C}^{-1}\mathbf{H}$ is also full rank and, with $\mathbf{B} = (\mathbf{H}^\top\mathbf{C}^{-\top}\mathbf{Y}\mathbf{C}^{-1}\mathbf{H})^{-1}$, $(\mathbf{B},\boldsymbol{\omega})$  is feasible for \eqref{eq:OCI_Kahan_optimal_Rz}.

\bibliographystyle{IEEEtran}
\bibliography{references-c,bibliography}

\end{document}